\documentclass[prl, aps,letterpaper, twocolumn, preprintnumbers,superscriptaddress]{revtex4}
\usepackage{enumerate}
\usepackage{verbatim}
\usepackage{amsmath}
\usepackage{latexsym}
\usepackage{revsymb}
\usepackage{ifthen}
\usepackage{mathrsfs}

\usepackage{natbib}
\usepackage{amsfonts}
\usepackage{amsmath}
\usepackage{amssymb}
\usepackage{amsthm}
\usepackage{graphicx}

\newcommand{\be}{\begin{eqnarray} \begin{aligned}}
\newcommand{\ee}{\end{aligned} \end{eqnarray} }
\newcommand{\benn}{\begin{eqnarray*} \begin{aligned}}
\newcommand{\eenn}{\end{aligned} \end{eqnarray*} }

\newcommand*{\textfrac}[2]{{{#1}/{#2}}}

\newcommand*{\cA}{\mathcal{A}}

\newcommand*{\cH}{\mathcal{H}}

\newcommand*{\cS}{\mathcal{S}}

\newcommand*{\cX}{\mathcal{X}}

\newcommand*{\tr}{\mathop{\mathrm{tr}}\nolimits}

\newcommand*{\supp}{\mathrm{supp}}

\newcounter{protoCount}
\newcounter{protoList}
\newsavebox{\tmpbox}
\newlength{\protobox}

\newcommand{\bc}{\begin{center}}
\newcommand{\ec}{\end{center}}

\newcommand{\id}{\mathbb{I}}
\newcommand{\idchannel}{\mathsf{id}}

\newcommand{\Tr}{\mathop{\mathrm{tr}}\nolimits}

\newtheorem{theorem}{Theorem}[section]
\newtheorem{lemma}[theorem]{Lemma}

\newcommand{\hil}{\mathcal{H}}

\usepackage{amsfonts}

\def\Complex{\mathbb{C}}

\def\id{\mathbb{I}}

\def\01{\{0,1\}}

\newcommand{\ket}[1]{|#1\rangle}

\newcommand{\outp}[2]{|#1\rangle\langle#2|}
\newcommand{\proj}[1]{|#1\rangle\langle#1|}

\newcommand{\inp}[2]{\langle{#1}|{#2}\rangle} 

\newcommand{\states}{\mathcal{S}}
\newcommand{\bounded}{\mathcal{B}}
\newcommand{\hout}{\mathcal{H}_{\rm out}}
\newcommand{\hin}{\mathcal{H}_{\rm in}}

\newcommand{\setX}{\mathcal{X}}

\newcommand{\assign}{:=}

\newcommand*{\mr}[1]{\mu_{\alpha,Q}(#1)}
\newcommand{\capa}{\chi^*}

\newcommand{\setA}{\mathcal{A}}

\begin{document}

\title{A strong converse for classical channel coding using entangled inputs}
\author{Robert \surname{K{\"o}nig}}
\email[]{rkoenig@caltech.edu}
\affiliation{Institute for Quantum Information, Caltech, Pasadena CA 91125, USA}
\author{Stephanie \surname{Wehner}}
\email[]{wehner@caltech.edu}
\affiliation{Institute for Quantum Information, Caltech, Pasadena CA 91125, USA}
\begin{abstract}
A fully general strong converse for channel coding states that when the rate of sending classical information exceeds the capacity of a quantum channel,  the probability of correctly decoding goes to zero exponentially in the number of channel uses, even when we allow code states which are entangled across several uses of the channel.  Such a statement was previously only known for classical channels and the quantum identity channel. By relating the problem to the additivity of minimum output entropies, we show that a strong converse holds for a large class of channels, including all unital qubit channels, the $d$-dimensional depolarizing channel and the Werner-Holevo channel. This  further justifies the interpretation of the classical capacity as a sharp threshold for information-transmission.
\end{abstract}
\maketitle

A fundamental problem in quantum information theory is the transmission of classical information over (noisy) quantum channels. 
As a simple example, suppose we send $M$~classical bits
 using a qubit identity channel $n$~times. Clearly~\footnote{By coding into orthogonal states.}, this can be done reliably if  $M\leq n$, but if the number of classical bits exceeds the number of qubits sent ($M>n$) we are no longer able to recover the encoded information with perfect accuracy~\footnote{This statement can be seen as a special case of  Holevo's channel coding theorem~\cite{holevo:bound}.}. 
This situation is  analogous to the problem of information transmission over a noisy classical channel.
Here, there exists a constant~$C$, called the {\em classical capacity}, 
which determines the maximal number of classical bits that can be sent reliably per channel use: 
by using the channel $n$~times, we can reliably transmit $M$~bits if and only if the {\em rate} $R=\frac{M}{n}$ satisfies $R\leq C$ in the asymptotic limit. This is known as the {\em coding theorem} due to Shannon~\cite{shannon:coding}.
For example, for the binary bit flip channel, which flips an input bit with probability $p$, this constant is given by $C= 1 - h(p)$, where $h$ is the binary entropy function. 
The unifying concept for both scenarios is that of the classical capacity~$C$. For the qubit identity channel 
Holevo's seminal result~\cite{holevo:bound} shows that the classical capacity is equal to~$1$.

In fact, for both the qubit identity channel and any classical channel, the classical capacity~$C$ imposes a 
sharp bound on our ability to recover classical information sent over the channel:
On the one hand if $R \leq C$, then it is possible to send $nR$~classical bits by using the channel $n$ times 
in such a way that the probability $P_{succ}$ of successful 
decoding goes to $1$ exponentially as $n\rightarrow \infty$. This is also referred to as the {\em achievability of the capacity}.  On the other hand, if $R > C$, then for any encoding and decoding scheme, $P_{succ}$ is exponentially small in the difference $n(R-C)$. This is referred to as the {\em strong converse} of the coding theorem for these channels. 

For classical noisy channels, the strong converse was established by Wolfowitz~\cite{wolfowitz}. For the qubit identity channel $\idchannel_2\equiv \idchannel_{\bounded(\mathbb{C}^2)}$, 
the argument is rather simple:
Suppose we encode a uniformly distributed $nR$-bit string $X\in\{0,1\}^{nR}$ using a family of $2^{nR}$ states $\{\rho_x\}_{x=1}^{2^{nR}}$ on $(\mathbb{C}^2)^{\otimes n}$ (i.e., of $n$~qubits). Then, for any decoding POVM $\{E_x\}_{x=1}^{2^{nR}}$ on $(\mathbb{C}^2)^{\otimes n}$, the average success probability of correctly decoding is bounded by
\begin{align*}
P_{succ}^{\idchannel_2}(n,R) &=\frac{1}{2^{nR}}\sum_x \tr(E_x\rho_x)\leq \frac{1}{2^{nR}}\sum_x \tr(E_x)\\
&=2^{-n(R-1)}\ .
\end{align*}
Here, we used the operator inequality $\rho_x\leq \id_{(\mathbb{C}^2)^{\otimes n}}$ for every $x$, and the fact that the operator elements of a POVM sum to the identity.
Due to the strong converse property, we can regard the capacity~$C$ as an exact measure of the information-carrying power of any classical channel and the quantum identity channel. 

Unfortunately, this appealing operational interpretation of the classical capacity~$C$ is not quite as complete for general quantum channels. While the achievability of the capacity has been established in~\cite{holevo:achieve,schumacher:achieve} (building on~\cite{hausladen:achieve}), only a {\em weak converse} has been shown without assumptions~\cite{holevo:bound}.
It merely states that for rates $R>C$ above the capacity, the success probability is bounded away from $1$. This is in contrast to a strong converse, which shows that this probability goes to zero exponentially, in the limit as~$n$ goes to infinity.

Here, we are interested in the validity of the strong converse property for a general quantum channel. Establishing such a converse is more difficult than for classical channels for the same reason it is difficult to compute the classical capacity of a quantum channel: We have to take into account the possibility that entanglement over several uses of the channel may help to increase the probability of successful decoding. Indeed, a recent breakthrough result by Hastings~\cite{hastings:additivity} shows that using entangled states can be advantageous. Formally, this is expressed by the product-state capacity~$C^{prod}_\Phi$: This is defined in the same way as the capacity, but with the restriction that the input states to the 
channel $\Phi^{\otimes n}$ have to be of tensor product form. Hasting's result 
shows that there are channels~$\Phi$ with $C^{prod}_\Phi <C_\Phi$.

In light of the advantage of entanglement for coding, it is natural to ask whether entanglement may invalidate the strong converse property: In particular, we study whether allowing arbitrary (entangled) input states does not affect the exponential decay of the success probability. Previous studies of  the region $R>C_\Phi$ were restricted to the case where the inputs are not entangled across different uses of the channel~\cite{ogawa:converse,winter:converse}, and are thus conceptually similar to the study of the achievability of
the product state capacity $C^{prod}_\Phi$ instead of the more general~$C_\Phi$.

\section{Main Result} 
Here, we prove a strong converse for a large number of quantum channels $\Phi$.
In particular, our result applies to 
\begin{enumerate}[(i)]
\item\label{it:first}
the qudit depolarizing channel 
\begin{align}
\Delta_r(\rho) &=r\rho+ (1-r)\frac{\id}{d}\ ,\label{eq:depolarizingchannel}
\end{align}
replacing any input state with the fully mixed state with probability $(1-r)$ for $-1/(d^2-1) \leq r \leq 1$,
\item\label{it:second}
any unital qubit channel~\footnote{A unital channel
maps the completely mixed state on $\cH_{in}$ to the completely mixed state on $\cH_{out}$.}, and, more generally, 
\item\label{it:third} any channel
which has additive minimum output $\alpha$-entropy $S^{\min}_\alpha$ for $\alpha \geq 1$ (close to $1$) as defined below~\footnote{The additivity property~$S_\alpha^{\min}(\Phi^{\otimes n}) = n \cdot S_\alpha^{\min}(\Phi)$ is equivalent to the multiplicativity of the  maximum output $\alpha$-norm defined as $\nu_\alpha(\Phi) = \max_\rho \|\Phi(\rho)\|_\alpha$, where $\|A\|_\alpha=\left(\tr|A|^\alpha\right)^{1/\alpha}$.},
and the following covariance property: 
there is a pair of unitary representations of some group $G$ on the input space $\cH_{in}$ and the output space $\cH_{out}$, 
respectively, such that
\begin{align*}
g \Phi(\rho) g^\dagger &= \Phi(g\rho g^\dagger)\qquad\textrm{ for all }g\in G\ ,
\end{align*}
where the representation on $\cH_{out}$ is irreducible. An example of such a channel is the Werner-Holevo channel~\cite{werner:channel}.
\end{enumerate}

More formally, we are concerned with (noisy) quantum channels, i.e., completely positive trace-preserving maps (CPTPM) $\Phi: \bounded(\hin) \rightarrow \bounded(\hout)$. Throughout, we restrict our attention to finite-dimensional Hilbert spaces $\hin$ and $\hout$. A code of rate $R$ for $\Phi$  specifies (for every~$n$) a family $\{\rho_x\}_{x=1}^{2^{nR}}$ of states on $\hin^{\otimes n}$, where
 $\rho_x$ is the quantum codeword associated with the classical  message $x \in \{1,\ldots,2^{nR}\}$. A corresponding decoder is  a POVM $\{E_x\}_{x=1}^{2^{nR}}$ on $\hout^{\otimes n}$.
We are interested in the average success probability of decoding correctly, that is, the quantity
\begin{align}\label{eq:success}
P_{succ}^\Phi(n,R) &=\frac{1}{2^{nR}} \sum_{x=1}^{2^{nR}}\tr(E_x\Phi^{\otimes n}(\rho_x))\ .
\end{align} 
In this terminology, we  show the following:

\bigskip
\textbf{Theorem} 
\emph{Let $\Phi$ be a CPTPM described by (i)--(iii), and let $C_\Phi$ be its classical capacity. There exists a constant $\gamma>0$ such that the following holds: For any code of rate $R$, and any corresponding decoder, the success probability $P_{succ}^\Phi(n,R)$ is upper bounded by $2^{-\gamma\cdot n (R-C_\Phi)}$ (for sufficiently large $n$).}
 
\bigskip
Thus the success probability decays exponentially when coding at rates above the capacity.

\bigskip
\textbf{Background} Before giving a short overview of our proof, let us briefly recall how the study of the achievability of rates  below the capacity can be subdivided into three major components: one begins by setting up a connection between the operational problem of coding and an entropic quantity. More precisely, one can show that there exists codes such that
the success probability  has a behavior of the form
\begin{align}
P^{\Phi}_{succ}(n,R)&=1-e^{-n\delta  (\bar{\chi}^*(\Phi)-R)}\ ,\label{eq:psuccessbelowcapacity}
\end{align}
with $\delta \geq 0$ for rates $R$  smaller than
\begin{align}\label{eq:barchidef}
\bar{\chi}^*(\Phi)\assign\lim_{n\rightarrow \infty}\frac{1}{n}\chi^*(\Phi^{\otimes n})\ .
\end{align}
This quantity is the {\em regularized} version of the Holevo-quantity of the channel $\Phi$, i.e.,
\begin{align}\label{eq:unregularized}
\chi^*(\Phi) &\assign \max_{\{p_x,\rho_x\}_x}
\chi(\{p_x,\Phi(\rho_x)\}_x)\ ,
\end{align}
which in turn is defined in terms of the  Holevo quantity of an ensemble $\{p_x,\sigma_x\}_x$, given by
\begin{align}
\chi(\{p_x,\sigma_x\}_x) \assign S\left(\sum_x p_x\sigma_x\right)-\sum_x p_xS(\sigma_x)\ .\label{eq:holevoquantity}
\end{align}
 This is the first step in the study of the coding problem. It reduces the operational problem of coding to the study of the quantity~\eqref{eq:barchidef}. In particular,~\eqref{eq:psuccessbelowcapacity} tells us that we can code with exponentially small error at any rate $R<\bar{\chi}^*(\Phi)$.

The second component is to study general properties of the quantity $\bar{\chi}^*(\Phi)$. The computation of this value is drastically simplified in cases where the Holevo quantity is \emph{additive}, that is, 
\begin{align}
\chi^*(\Phi^{\otimes n-1}\otimes\Phi) &=\chi^*(\Phi^{\otimes n-1})+\chi^*(\Phi)
\label{eq:additivity}
\end{align}
for all $n > 1$, since this implies $\bar{\chi}^*(\Phi)=\chi^*(\Phi)$.  Note
that part of this statement, the so-called {\em subadditivity}
\begin{align*}
\chi^*(\Phi^{\otimes n-1}\otimes\Phi) &\geq \chi^*(\Phi^{\otimes n-1})+\chi^*(\Phi)\ ,
\end{align*}
is trivial, as it corresponds to restricting to product states. Showing whether or not~\eqref{eq:additivity} holds for a given channel $\Phi$ is 
a called an  {\em additivity problem}. It has several equivalent formulations: for example, the quantity  $\chi^*(\Lambda)$, for any CPTPM~$\Lambda$,
 can be reexpressed in terms of the relative entropy $D$ as
\begin{align}
\chi^*(\Lambda) &=\min_\sigma\max_\rho D(\Lambda(\rho)\|\Lambda(\sigma))\label{eq:accrelative}
\end{align}
as shown in~\cite{schumacher:capacityVSrelative}. The physical significance of the additivity property~\eqref{eq:additivity} stems from the fact that~\eqref{eq:barchidef} is a formula for the capacity $C_\Phi$, while~\eqref{eq:unregularized} is equal to the product state capacity $C^{prod}_\Phi$~\footnote{In addition to~\eqref{eq:psuccessbelowcapacity}, Fano's inequality gives a (weak) converse involving the quantity~\eqref{eq:barchidef}, which shows that $C_\Phi=\bar{\chi}^*(\Phi)$.}. Additivity of $\chi^*$ for a channel~$\Phi$ therefore implies that there is no advantange in using entangled states for coding in the asymptotic limit. 

Finally, one needs to investigate the additivity problem (cf.~\eqref{eq:additivity}), which is poorly understood in general. King~\cite{king:depol} has shown additivity of $\chi^*$ for the depolarizing channel~\eqref{eq:depolarizingchannel}. His proof 
uses the fact that for any covariant channel $\Phi$, the Holevo quantity is related to the 
minimum output entropy~\cite{holevo:covariant} 
\begin{align}
S^{\min}(\Phi) &\assign\min_\rho S(\Phi(\rho))\label{eq:minimumoutputentropy}
\end{align}
by
\begin{align}
\chi^*(\Phi) &=\log d_{out}-S^{\min}(\Phi)\ ,\label{eq:chiexpr}
\end{align}
where $d_{out}$ is the dimension of the output space $\hout$. King then establishes the additivity of $S^{\min}$ for the depolarizing channel $\Delta_r$ by showing that the related minimum $\alpha$-R\'enyi-entropies $S_\alpha^{\min}$ (defined below) are additive
for~$\Delta_r$. This implies additivity of $\chi^*$, and leads to an explicit formula for the capacity $C_{\Delta_r}$.

\bigskip
\textbf{Proof outline} Our approach to coding at rates above the capacity has the same overall structure as the study of the achievability explained above. The strong converse theorem is obtained by (a)~relating the decoding probability to entropic quantities, (b)~rephrasing the resulting additivity problems and finally (c)~showing that the channels~(i)--(iii) satisfy these additivity properties. 

The relevant quantities in our case turn out to be the following R\'enyi-entropic versions of the above quantities.
For $\alpha \geq 1$, we use~\footnote{Note that the expression $\sigma^{1-\alpha}$ in $D_\alpha(\rho\|\sigma)$ requires to augment definition~\eqref{eq:dalphadef}: We will only invert $\sigma$ on its support~$\supp(\sigma)$, and set $D_\alpha(\rho\|\sigma)=\infty$ if $\supp(\sigma)\not\subset\supp(\rho)$. }
\begin{align}
S_\alpha(\rho) &\assign \frac{1}{1-\alpha}\tr(\rho^\alpha)\nonumber\\
D_\alpha(\rho\|\sigma)&\assign \frac{1}{\alpha-1}\log\tr(\rho^\alpha\sigma^{1-\alpha})\label{eq:dalphadef}\\
\chi_\alpha(\{p_x,\sigma_x\}_x) &\assign \frac{\alpha}{\alpha-1} \log\tr\left(\sum_x p_x \sigma_x^\alpha\right)^{1/\alpha}\ .\nonumber
\end{align}
We also need the corresponding derived quantities $\chi^*_\alpha(\Phi)$, $\bar{\chi}^*_\alpha(\Phi)$, and $S^{\min}_\alpha(\Phi)$ defined as in~\eqref{eq:unregularized},~\eqref{eq:barchidef} and~\eqref{eq:minimumoutputentropy}, respectively.

We now give a sketch of the proof, following the three steps (a)--(c) outlined above (Details can be found in the appendix). First, 
we relate our operational problem 
to the regularized quantity $\bar{\chi}_\alpha^*(\Phi)$ by  showing that for any code of rate $R$, we have 
\begin{align} \tag{\ref{eq:psuccessbelowcapacity}$^\prime$}
P_{succ}^\Phi(n,R)\lesssim 2^{-n(1-\frac{1}{\alpha})(R-\bar{\chi}_\alpha^*(\Phi))}\textrm{ for all }\alpha\geq 1\ \label{eq:psuccupperbound}
\end{align}
for sufficiently large $n$. This is the analog of~\eqref{eq:psuccessbelowcapacity}. It shows that for any rate $R>\bar{\chi}_\alpha^*(\Phi)$, the success probability decays exponentially with~$n$.

Clearly, the quantity~$\bar{\chi}_\alpha^*(\Phi)$ again has a particularly simple form if  $\chi_\alpha^*$ is additive as in~\eqref{eq:additivity}. To study additivity of the quantity $\chi_\alpha^*(\Phi)$, the second step of our proof is to derive the following analog of~\eqref{eq:accrelative}, essentially following the steps of Schumacher and 
Westmoreland~\cite{schumacher:capacityVSrelative}
\begin{align}
\min_{\sigma_{out}}\max_{\rho} D_\alpha(\Lambda(\rho)\|\sigma_{out})&\leq \chi_\alpha^*(\Lambda)\nonumber\\
&\leq \min_{\sigma_{in}}\max_\rho D_\alpha(\Lambda(\rho)\|\Lambda(\sigma_{in}))\ . \tag{\ref{eq:accrelative}$^\prime$}\label{eq:chidmod}
\end{align}
As before, additivity of the quantity $\chi_\alpha^*(\Phi)$ is intimately connected to the classical capacity $C_\Phi$: As shown by Ogawa and 
Nagaoka~\cite{ogawa:converse}, for every $\varepsilon>0$, we have  $\chi_\alpha^*(\Phi)<C_\Phi+\varepsilon$ for all~$\alpha\geq 1$ in some neighborhood of~$1$. In particular, with~\eqref{eq:psuccupperbound}, this shows that additivity of~$\chi_\alpha^*$ for all~$\alpha$ in the vicinity of~$1$ implies a strong converse, that is, an exponential decay of the success probability for any rates~$R>C_\Phi$. 
Since it is known~\cite{ogawa:converse,winter:converse} that coding with product states at rates above the capacity leads to the same exponential behavior, we can conclude 
that entanglement provides no operational advantage. 

Finally, we show additivity of $\chi_\alpha^*$ for the special class of channels $\Phi$  satisfying our assumptions (i)--(iii). For these channels,  the covariance properties imply that both the lower and upper bound in~\eqref{eq:chidmod} coincide and are attained when $\sigma_{in}$ and $\sigma_{out}$ are completely mixed. By definition, this means that these channels satisfy the R\'enyi-entropic version 
\begin{align} \tag{\ref{eq:chiexpr}$^\prime$}
\chi_\alpha^*(\Phi) &=\log d_{out}-S_\alpha^{\min}(\Phi)\ \label{eq:chiexprx}
\end{align} 
of~\eqref{eq:chiexpr}.
Additivity of $\chi_\alpha^*$ is shown by combining~\eqref{eq:chidmod} with~\eqref{eq:chiexprx}, as follows.  For $\sigma_{in} = \id/d$ equal to the fully mixed state, we get 
\begin{align}
\chi_\alpha^*(\Phi^{\otimes n})&\leq \max_\rho D_\alpha(\Phi^{\otimes n}(\rho)\|
\Phi^{\otimes n}((\id/d_{in})^{\otimes n}))\nonumber\\
&=\log d_{out}^n-S_\alpha^{\min}(\Phi^{\otimes n})\nonumber\\
&=n\log d_{out}- n\cdot S_\alpha^{\min}(\Phi)\ .\label{eq:upperboundchialpha}
\end{align}
In the last step,  we used the additivity of the minimum output $\alpha$-entropy $S_\alpha^{\min}$ for the channels of interest for $\alpha\geq 1$ 
close to $1$ (cf.~\cite{king:unital} for qubit unital channels, \cite{king:depol} for the depolarizing channel, and \cite{datta:werner,alicki:werner,yura:werner} for the Werner-Holevo 
channel). By the subadditivity property of the quantity $\chi_\alpha^*$, we know that $n\chi_\alpha^*(\Phi)\leq \chi_\alpha^*(\Phi^{\otimes n})$. Combining this with~\eqref{eq:chiexprx} and~\eqref{eq:upperboundchialpha} proves additivity, that is, $\bar{\chi}_\alpha^*(\Phi)=\chi_\alpha^*(\Phi)=\log d_{out}-S_{\min}(\Phi)$. This concludes the proof of our main result.
\section{Conclusion}
In summary, we have shown that for a large class of practically relevant quantum channels, the probability of reliably transmitting $nR$~classical bits by $n$~uses of the channel has an asymptotic behavior of the form $2^{-\gamma n(R-C)}$ for some constant~$\gamma>0$ when coding at rates~$R$ above the classical capacity $C$. Such a statement was previously only known for classical channels and the identity channel. 
Our result has direct practical applications
to quantum cryptography, especially in the so-called noisy-quantum-storage model~\cite{prl:noisy, noisy:robust}, where the adversary is restricted to using low-capacity channels. For these applications, some knowledge about the optimal constant~$\gamma$ will be useful. Our work provides bounds on this value, about which little is known 
even in the classical case.  

On a more fundamental level, our result implies that for the quantum channels considered, using entanglement provides no advantage in all rate regimes. These channels therefore behave just as classical channels with respect to the transmission of classical information.  
Establishing strong converses for a wider class of channels is of fundamental importance, as this is the natural counterpart of the achievability statement of the capacity. Of particular interest in this context are channels whose Holevo-quantity is non-additive~\cite{hastings:additivity}.  While we do not explicitly use this fact, the Holevo-quantity is additive for the channels considered in this paper.

Showing that  the success probability of decoding has an exponential behavior both below and above the capacity confirms our interpretation of the classical capacity as the single relevant measure of the usefulness of a quantum channel for classical communication.

\acknowledgments
We acknowledge support by  NSF grants PHY-04056720 and PHY-0803371.

\newpage

\setcounter{section}{1}

In this appendix, we provide a detailed proof of the strong converse theorem for all channels described by~\eqref{it:first}--\eqref{it:third}. Let us first argue that it suffices to consider channels of the type~\eqref{it:third}, i.e., covariant channels. Indeed, 
the $d$-dimensional depolarizing channel~\eqref{eq:depolarizingchannel} is just a special example of~\eqref{it:third}, since it has additive minimum 
output $\alpha$-entropy~\cite{king:depol} and is covariant with respect to the unitary group. For unital qubit channels, first observe that the quantity~$\bar{\chi}_\alpha^*(\Phi)$ of interest  remains unchanged when 
considering a unitarily equivalent channel, i.e., one which additionally 
conjugates  the input and output with fixed unitaries $U_{in}$ and $U_{out}$, respectively. 
It has been shown~\cite{ruskai:qubit,datta:outputPurity} that any one qubit unital channel is unitarily equivalent to a Pauli diagonal channel
\begin{align*} 
\mathcal{F}(\rho) &=\sum_{j=0}^3 \alpha_j \sigma_j \rho \sigma_j
\end{align*}
for some $\alpha_j\geq 0$, where $\{\sigma_j\}_j$ are the Pauli matrices.  This channel is an instance of~\eqref{it:third}, since it has additive 
minimum output $\alpha$-entropy~\cite{king:unital} and is invariant with respect to the irreducible action of the Pauli group on $\mathbb{C}^2$. We will therefore restrict our attention to covariant channels in this appendix.
We now state our main result more formally:

\begin{theorem}[Strong converse]\label{thm:maintheorem}
Let $\Phi: \bounded(\hin) \rightarrow \bounded(\hout)$ be a CPTPM satisfying
\begin{enumerate}[(A)]
\item\label{eq:covarianceproperty}
$\Phi$ is covariant with respect to a pair of unitary representations of a compact group $G$ on $\hin$ and $\hout$, where the representation on $\hout$ is irreducible.
\item\label{eq:minimumoutputproperty}
The minimum output entropy $S_\alpha^{\min}(\Phi)$ is additive for $\alpha \geq 1$ (for $\alpha$ close to 1).
\end{enumerate}
Then the strong converse holds for $\Phi$, that is, $P_{succ}^\Phi(n,R)\rightarrow 0$ exponentially for 
any rate $R>\bar{\chi}^*(\Phi)=C_\Phi$.
\end{theorem}
We assume throughout that the representations of $G$ are continuous, and state our proofs for the case where $G$ is finite (the general case is analogous, see e.g.,~\cite{holevo:covariant} for details).

For rates $R < C_\Phi$, the rate of convergence 
of $P_{succ}^{\Phi}(n,R)\rightarrow 1$ for the optimal code and decoder is measured by the so-called reliability rate function (see e.g.,~\cite{burnashev:reliability})
\begin{align}
E^\Phi(R)=\lim_{n\rightarrow \infty } \sup\frac{-\log (1-P_{succ}^{\Phi}(n,R))}{n}\ . \label{eq:reliabilityfunction}
\end{align}
For $R > C_\Phi$, we are interested in the rate at which $P^{\Phi}_{succ}(n,R)\rightarrow 0$ as $n\rightarrow \infty$. In analogy to~\eqref{eq:reliabilityfunction}, we introduce the function
\begin{align}
E^{\Phi}(R)=\lim_{n\rightarrow \infty } \inf\frac{-\log P_{succ}^{\Phi}(n,R)}{n}\ .\label{eq:errorexponent}
\end{align}

We now make the three main steps $(a)$--$(c)$ in the proof of our theorem more explicit.   The following lemma gives a bound on~\eqref{eq:errorexponent} in 
terms of the regularized $\alpha$-Holevo quantity, and thus connects $\alpha$-Holevo quantities to the operational coding problem.
\begin{lemma}\label{lem:errorexponentbound}
For all CPTPMs $\Phi: \bounded(\hin) \rightarrow \bounded(\hout)$ 
\begin{enumerate}
\item\label{it:firstpartoperational}
The operational quantity~\eqref{eq:errorexponent} is bounded by the regularized $\alpha$-Holevo quantity as
\begin{align*}
E^{\Phi}(R)\geq \left(1-\frac{1}{\alpha}\right)\left(R-\bar{\chi}^*_\alpha(\Phi)\right)\qquad\textrm{ for all }\alpha > 1\ .
\end{align*}
\item\label{it:secondpartoperational}
For every $R>\chi^*(\Phi)$, there exists $\beta=\beta(R) > 1$ such that $R>\chi_\alpha^*(\Phi)$ for all $1<\alpha< \beta$.
\end{enumerate}
In particular, $E^{\Phi}(R)>0$ for all $R>\chi^*(\Phi)$ if $\chi_\alpha^*$ is additive for~$\Phi$ for all $\alpha > 1$ close to $1$.
\end{lemma}
\noindent The proof of this lemma, which is essentially identical to a derivation in~\cite{ogawa:converse}, is given in Appendix~\ref{lem:errorexponentbound}. Note that for the channels of interest, we have $C_\Phi=\chi^*(\Phi)$. Therefore, Lemma~\ref{lem:errorexponentbound} reduces the problem of establishing a strong converse to the additivity of $\bar{\chi}_\alpha^*$.

Recall that the second step is to bound the quantity $\chi_\alpha^*$ in terms of a generalized form of the relative entropy for $\alpha>1$.
\begin{lemma}\label{lem:evaluation}
Let $\Lambda: \bounded(\hin) \rightarrow \bounded(\hout)$ be a CPTPM, and $\alpha > 1$. The quantity $\chi_\alpha^*(\Lambda)$ is related to $D_\alpha$ by
\begin{align}
\min_{\sigma_{\rm out}} \max_{\rho} D_\alpha(\Lambda(\rho)\|\sigma_{\rm out})&\leq \chi_\alpha^*(\Lambda)\nonumber\\
&\leq \min_{\sigma_{\rm in}}\max_\rho D_\alpha(\Lambda(\rho)\|\Lambda(\sigma_{\rm in}))\ .\label{eq:chiD}
\end{align}
Moreover, if  $\Phi: \bounded(\hin) \rightarrow \bounded(\hout)$  is a CPTPM satisfying the covariance 
property~\eqref{eq:covarianceproperty}, then 
\begin{align}
\chi_\alpha^*(\Phi)&=\min_{\sigma_{\rm out}} \max_{\rho} D_\alpha(\Phi(\rho)\|\sigma_{\rm out})\nonumber\\
&=\min_{\sigma_{\rm in}}\max_\rho D_\alpha(\Phi(\rho)\|\Phi(\sigma_{\rm in}))\nonumber\\
&=\log d_{out}-S_\alpha^{\min}(\Phi)\ ,\label{eq:covariantchialpha}
\end{align}
where $d_{out}$ is the dimension of $\hout$.
\end{lemma}
\begin{proof}
The inequalities~\eqref{eq:chiD} follow from a more general statement shown in Lemma~\ref{lem:generalevaluation}. 
We now show that identity~\eqref{eq:covariantchialpha} follows from~\eqref{eq:chiD} and a straightforward application of the covariance property: 
Consider a CPTPM $\Phi:\bounded(\hin)\rightarrow\bounded(\hout)$ with property~\eqref{eq:covarianceproperty}, and let $\Lambda:\bounded(\widetilde{\hin})\rightarrow\bounded(\hout)$ be a unital CPTPM with the same range. 
Fix some states $\rho' \in \states(\hin)$ and $\sigma \in \states(\widetilde{\hin})$ and observe that 
\begin{align*}
 \max_\rho\tr (\Phi(\rho)^{\alpha}\Lambda(\sigma)^{1-\alpha}) &\geq \tr(\Phi(g\rho'g^\dagger)^{\alpha}\Lambda(\sigma)^{1-\alpha})\\
&= \tr(\Phi(\rho')^{\alpha}g^\dagger \Lambda( \sigma)^{1-\alpha}g) 
 \end{align*}
for all $g\in G$. Here we used the
 covariance of $\Phi$, $(gAg^\dagger)^\beta=g A^\beta g^\dagger$
and the cyclicity of the trace in the last identity.
Taking the average over all $g\in G$
gives
 \begin{align*}
 &\max_\rho \tr(\Phi(\rho)^\alpha\Lambda(\sigma)^{1-\alpha})\\
&\phantom{====}\geq
 \tr\left(\Phi(\rho')^\alpha \frac{1}{|G|}\sum_{g\in G} g^\dagger \Lambda(\sigma)^{1-\alpha}g\right)\\
 &\phantom{====} =\tr(\Phi(\rho')^\alpha)\frac{1}{d_{out}}\cdot\sum_i \lambda_i^{1-\alpha}\ ,
  \end{align*}
where $\{\lambda_i\}$ are the (non-zero) eigenvalues of the operator $\Lambda(\sigma)$. Note that for $\alpha > 1$ the expression $f(\lambda) \assign \sum_i \lambda_i^{1-\alpha}$ is minimal if $\lambda_i=1/d_{out}$ for all $i=1,\ldots,d_{out}$: this follows because $f(\lambda_1,\lambda_2,\ldots)\geq f(\frac{\lambda_1+\lambda_2}{2},\frac{\lambda_1+\lambda_2}{2},\ldots)$ by the convexity of the function $x\mapsto x^{1-\alpha}$, and the symmetry of $f$ with respect to permutations of its arguments. 
This minimum is attained if $\Lambda(\sigma)$ is completely mixed, or (since $\Lambda$ is unital) by choosing $\sigma=\id/\tilde{d}_{in}$  to be the fully mixed state on $\widetilde{\hin}$.
We conclude that for all~$\rho'$ and~$\sigma$ 
  \begin{align*}
 \max_\rho \tr(\Phi(\rho)^{\alpha}\Lambda(\sigma)^{1-\alpha}) \geq \tr\left(\Phi(\rho')^{\alpha}\Lambda(\id/\tilde{d}_{in})^{1-\alpha}\right)\ .
 \end{align*}
 Hence, taking the maximum over $\rho'$ and the minimum over $\sigma$ gives
 \begin{align*}
 \min_\sigma\max_\rho D_\alpha\left(\Phi(\rho)\|\Lambda(\sigma)\right)=\max_\rho D_\alpha\left(\Phi(\rho)\|\Lambda\left(\textfrac{\id}{\tilde{d}_{in}}\right)\right)\ .
 \end{align*}
  Applying this equation to the cases~$\Lambda=\Phi$ (any covariant channel is unital),
 and $\Lambda=\idchannel$ equal to the identity channel on $\hout$ immediately shows that the two quantities in~\eqref{eq:covariantchialpha} are indeed equal, and given by
 \begin{align*}
& \max_\rho D_\alpha\left(\Phi(\rho)\|\textfrac{\id}{d_{out}}\right) \\
&\phantom{====} =\max_\rho \frac{1}{\alpha-1}\log d_{out}^{\alpha-1}\tr(\Phi(\rho)^\alpha)\\
 &\phantom{====}=\log d_{out}-S_{\alpha}^{\min}(\Phi)\ ,
 \end{align*} 
 as claimed. 
 \end{proof}

The last step in the proof of Theorem~\ref{thm:maintheorem} is to combine  
Lemma~\ref{lem:errorexponentbound} with the following statement derived in the main text.
\begin{theorem}[Additivity of $\chi_\alpha^*$]
Let $\Phi: \bounded(\hin) \rightarrow \bounded(\hout)$ be a CPTPM with properties~\eqref{eq:covarianceproperty} and~\eqref{eq:minimumoutputproperty} as in Theorem~\ref{thm:maintheorem}. Then for all $\alpha > 1$
\begin{align*}
\chi_\alpha^*(\Phi^{\otimes n}) &=n\cdot \chi^*_\alpha(\Phi)=n(\log d_{out}-S^{\min}_\alpha(\Phi))\ ,
\end{align*}
where $d_{out}$ is the dimension of $\hout$.
\end{theorem}

In the remainder of this appendix, we fill in the remaining technical details. In particular, we derive~\eqref{eq:chiD} of Lemma~\ref{lem:evaluation}, as well as Lemma~\ref{lem:errorexponentbound}. 

\section{$\alpha$-R{\'e}nyi quantities}\label{sec:dalphaProperties}
\setcounter{section}{2}
We begin with a few properties of $\alpha$-relative entropies.

\subsection{Properties of the $\alpha$-relative entropy}
If $\rho$ and $\sigma$ are classical (i.e., commuting), $D_\alpha(\rho\|\sigma)$ reduces to the classical $\alpha$-relative entropy defined 
in~\cite{csiszar:relative}. The quantity $D_\alpha(\rho\|\sigma)$ for $0\leq \alpha\leq 1$ was previously used, e.g., in~\cite{datta:relative,petz:alpha}.
Some of the following statements also hold for this regime, however, we concentrate on $\alpha>1$. 
We begin by showing positivity of $D_\alpha(\rho\|\sigma)$.

\begin{lemma}\label{lem:dalphaIsPositive}
$D_\alpha(\rho||\sigma) \geq 0$ for all states~$\rho, \sigma \in \states(\mathbb{C}^d)$ and $\alpha > 0$, where equality holds if
and only if $\rho = \sigma$.
\end{lemma}
\begin{proof}
Let $\lambda=(\lambda_1,\ldots,\lambda_d)$ and $\mu=(\mu_1,\ldots,\mu_d)$ 
denote the eigenvalues of $\rho$ and $\sigma$ (in some fixed order), respectively, and let $\lambda^\pi=(\lambda_{\pi(1)},\ldots,\lambda_{\pi(d)})$ be the reordered list, for every permutation $\pi\in S_d$.
Lemma~\ref{lem:majorize} implies that 
\begin{align}
\min_{\pi\in S_d}D_\alpha(\lambda^\pi\|\mu)\leq D_\alpha(\rho\|\sigma)\ .\label{eq:majorizationperm}
\end{align}
Inequality~\eqref{eq:majorizationperm} and the fact that
the classical $\alpha$-relative entropy
is non-negative~\cite{csiszar:relative} immediately imply that $D_\alpha(\rho\|\sigma)\geq 0$ for all $\rho$ and $\sigma$.

Note that the classical $\alpha$-relative entropy $D_\alpha(P\|Q)$ of two distributions $P$ and $Q$ vanishes only if $P\equiv Q$~\cite{csiszar:relative}. Combining this with~\eqref{eq:majorizationperm}, we conclude that if $D_\alpha(\rho\|\sigma)=0$, then $\rho$ and $\sigma$ must have the same spectrum $\lambda$ (up to some permutation). That is, there is a unitary $U$ such that $\sigma=U\rho U^\dagger$.  In particular, we get
\begin{align}
1&=2^{(\alpha-1)D_\alpha(\rho\|\sigma)}\nonumber\\
&=\sum_{i,j}\lambda_i^\alpha\lambda_j^{1-\alpha}|U_{ij}|^2\nonumber\\
&=\Lambda\star\Omega\ ,\label{eq:lambdastareq}
\end{align}
where $\Lambda$ and $\Omega$ are the matrices 
defined by $\Lambda_{ij}=\lambda_i^\alpha\lambda_j^{1-\alpha}$ and $\Omega_{ij}=|U_{ij}|^2$, and where we set
\begin{align*}
A \star B:=\sum_{ij} A_{ij} B_{ij}\ .
\end{align*}
Because $\Omega$ is a doubly stochastic matrix, it is a convex combination
\begin{align}
\Omega &=\sum_{\pi \in S_d} P(\pi) \pi\label{eq:birkhoffapplied}
\end{align}
of permutation matrices acting on $\mathbb{C}^d$, by Birkhoff's theorem (see e.g.,~\cite[Theorem 8.7.1]{horn&johnson:ma}).  From~\eqref{eq:lambdastareq} and~\eqref{eq:birkhoffapplied}, we get by linearity and the definition of $\Lambda$ 
\begin{align*}
1 &=\sum_{\pi\in S_d} P(\pi) \Lambda\star \pi\\
&=\sum_{\pi\in S_d} P(\pi) 2^{(\alpha-1)D_\alpha(\lambda\|\lambda^\pi)}
\end{align*}
With the positivity of the classical relative R\'enyi entropy, we conclude that
\begin{align*}
D_\alpha(\lambda\|\lambda^\pi)=0\ ,
\end{align*}
for every $\pi$ in the support of the distribution $P$. This in turn implies that for any such $\pi$, we have $\lambda=\lambda^\pi$.
In other words, only permutations $\pi$ which permute 
indices corresponding to a fixed eigenvalue among themselves appear in~\eqref{eq:birkhoffapplied}. We conclude that $\Omega$, and in particular $U$~are block-diagonal, with the different blocks corresponding to different eigenvalues, that is, we have
\begin{align*}
U=\bigoplus_{\lambda} U_{\lambda}\qquad \rho=\bigoplus_{\lambda} \lambda \id_{\mathbb{C}^{m_\lambda}}\ ,
\end{align*}
where the direct sums are over all distinct eigenvalues of $\rho$ and $m_\lambda$ is the multiplicity of $\lambda$. This shows that $\sigma=U\rho U^\dagger=\rho$ if $D_\alpha(\rho\|\sigma)=0$, as claimed.
\end{proof}

Next we consider the relative entropy $D_\alpha(\rho\|\sigma)$ for states $\rho,\sigma$ defined by ensembles: For an ensemble $\{p_x,\rho_x\}_{x\in\cX}$,
we introduce a corresponding  classical-quantum state (a \emph{cq-state})
\begin{equation}\label{eq:cqstate}
\rho_{XQ} = \sum_{x=1}^{|\cX|} p_x \underbrace{\outp{x}{x}}_X \otimes \underbrace{\rho_x}_Q \in \states(\hil_X \otimes \hil_Q)
\end{equation}
where $\{\ket{x}\}_{x=1}^{|\setX|}$ is an orthonormal basis of $\hil_X\cong\mathbb{C}^{|\cX|}$. Note that this defines a one-to-one correspondence between ensembles and cq-states. The following lemma shows how the relative entropy of a cq-state $\rho_{XQ}$  and a product state $\rho_{X}\otimes\sigma_Q$ decomposes into a sum of two terms. Only one of the terms depends on $\sigma_Q$. It is given by the relative entropy of $\sigma_Q$ and some state $\mu_{Q}=\mu_{\alpha,Q}(\rho_{XQ})$ which is defined in terms of the ensemble.

\begin{lemma}\label{lem:dalphaMinimized}
For every cq-state $\rho_{XQ} \in \states(\hil_X \otimes \hil_Q)$, define the state
$\mu_{Q}=\mr{\rho_{XQ}}\in\cS(\cH_Q)$ by
\begin{align*}
\mr{\rho_{XQ}} \assign 
\frac{1}{\tr(\zeta_Q)}\cdot\zeta_Q\textrm{ with }\zeta_Q=\left(\sum_x p_x \rho_x^\alpha\right)^{\textfrac{1}{\alpha}}\ .
\end{align*}
Then 
\begin{align*}
D_\alpha(\rho_{XQ}\|\rho_X\otimes\sigma_Q)=D_\alpha(\rho_{XQ}\|\rho_X\otimes \mu_Q)+D_\alpha(\mu_Q\|\sigma_Q)
\end{align*}
for all states $\sigma_Q \in \states(\hil_Q)$.
\end{lemma}
\begin{proof}
Observe that
\begin{align*}
D_{\alpha}(\rho_{XQ}\|\rho_X\otimes\sigma_Q)&=\frac{1}{\alpha-1} \log \tr\left(\sum_x p_x\rho_x^\alpha\sigma_Q^{1-\alpha}\right)\\
D_{\alpha}(\rho_{XQ}\|\rho_X\otimes\mu_Q)&=\frac{\alpha}{\alpha-1}\log\tr(\zeta_Q)\ .\nonumber
\end{align*}
In particular, we get
\begin{align*}
2^{D_\alpha(\rho_{XQ}\|\rho_X\otimes\sigma_Q)}&=\tr\left(\zeta_Q^{\alpha}\sigma_Q^{1-\alpha}\right)^{\frac{1}{\alpha-1}}\\
&=\tr(\zeta_Q)^{\frac{\alpha}{\alpha-1}}\tr\left(\mu_Q^\alpha\sigma_Q^{1-\alpha}\right)^{\frac{1}{\alpha-1}}\ ,
\end{align*}
from which the claim follows immediately.
\end{proof}

\subsection{Relating $D_\alpha$ to $\chi_\alpha$}
We are interested in  $\alpha$-Holevo quantities associated with ensembles $\{p_x,\rho_x\}_x$. Again, it is convenient to consider the corresponding cq-states~\eqref{eq:cqstate}. We define an $\alpha$-Holevo quantity of a cq-state as the corresponding quantity of the associated ensemble, that is,
\begin{align*}
\chi_\alpha(\rho_{XQ})& :=\chi_\alpha(\{p_x,\rho_x\}_x)\ ,
\end{align*}
and we will use ensembles and cq-states interchangeably.

We now essentially follow the arguments that Schumacher and Westmoreland~\cite{schumacher:capacityVSrelative} use to relate the Holevo quantity $\chi$ to the relative entropy $D$. Our goal is to obtain a similar characterization of $\chi_\alpha^*$ in terms of $D_\alpha$, as expressed by Lemma~\ref{lem:generalevaluation} below.
 As a first step, we express the quantity $\chi_\alpha(\rho_{XQ})$  by an optimization over relative $\alpha$-entropies of cq-states.
\begin{lemma}\label{lem:chialphadalpha}
For any cq-state $\rho_{XQ} \in \states(\hil_X \otimes \hil_Q)$,  we have the identity
\begin{align}
\chi_\alpha(\rho_{XQ})=\min_{\sigma_Q} D_\alpha (\rho_{XQ}\|\rho_X\otimes\sigma_Q)\label{eq:ten}
\end{align}
where the minimum is attained for the state $\sigma_Q = \mr{\rho_{XQ}}$ defined in Lemma~\ref{lem:dalphaMinimized}.
Furthermore, we have for all $\sigma_Q \in \states(\hil_Q)$
\begin{align}
D_\alpha(\rho_{XQ}\|\rho_X\otimes\sigma_Q)\geq \chi_\alpha(\rho_{XQ})\ \label{eq:nine}
\end{align}
with equality if and only if $\sigma_Q=\mr{\rho_{XQ}}$.
\end{lemma}
\begin{proof}
The fact that $\mr{\rho_{XQ}}$ achieves the minimum on the rhs.~of~\eqref{eq:ten} follows from Lemma~\ref{lem:dalphaMinimized} and the positivity of $D_\alpha$ shown in  Lemma~\ref{lem:dalphaIsPositive}. Inserting the definition of $\mr{\rho_{XQ}}$ into the expression on the rhs. of~\eqref{eq:ten} proves the validity of~\eqref{eq:ten}. Inequality~\eqref{eq:nine}  directly follows from~\eqref{eq:ten}, Lemma~\ref{lem:dalphaMinimized} and Lemma~\ref{lem:dalphaIsPositive}. 
\end{proof}
Note that by reinserting~\eqref{eq:ten} into the identity given in Lemma~\ref{lem:dalphaMinimized}, we obtain the identity
\begin{align}
\chi_\alpha(\rho_{XQ}) &+D_\alpha(\mr{\rho_{XQ}}\|\sigma_Q)=D_\alpha(\rho_{XQ}\|\rho_X\otimes\sigma_Q)\label{eq:eight}
\end{align}
for all $\sigma_Q\in\cS(\cH_Q)$. As a next step, we extend the cq-state $\rho_{XQ}$ by an additional classical symbol and show how $\chi_\alpha$ for the new state relates to the original quantity.
\begin{lemma}
Consider the cq-state $\rho_{XQ}$ of Eq.~\eqref{eq:cqstate} and let $\ket{\bot}$ be a normalized state on $\cH_X$ such that $\inp{\bot}{x} = 0$ for all $x \in \setX$.
Let $\rho_0 \in \states(\hil_Q)$ be arbitrary
 and consider the cq-state
$$
\rho_{XQ}'=(1-\eta)\rho_{XQ}+\eta \proj{\bot}\otimes\rho_0\ ,
$$
for some parameter $\eta\in[0,1]$. 
Then (for $\alpha > 1$)
\begin{align}\label{eq:fourteen}
\chi_\alpha(\rho_{XQ}')-\chi_\alpha(\rho_{XQ})\geq \eta \left(D_\alpha(\rho_0\|\mr{\rho_{XQ}'})-\chi_\alpha(\rho_{XQ})\right).
\end{align}
\end{lemma}
\begin{proof}
For simplicity, set $\sigma_Q'=\mu_{\alpha,Q}(\rho_{XQ}')$, where $\mu_{\alpha,Q}$ is defined as in Lemma~\ref{lem:dalphaMinimized}. For $\alpha>1$ we then have (by the first part of Lemma~\ref{lem:chialphadalpha})
\begin{align*}
&\chi_\alpha(\rho_{XQ}')=D_\alpha(\rho_{XQ}'\|\rho_X'\otimes\sigma'_Q)\\
&=\frac{1}{\alpha-1} \log\tr\left((1-\eta)\cdot \sum_x P_X(x)\rho_x^\alpha \sigma_Q'^{1-\alpha} +\eta\cdot \rho_0^\alpha\sigma_Q'^{1-\alpha}\right)\\
&\geq (1-\eta) \frac{1}{\alpha-1}\log \tr\left(\sum_x P_X(x)\rho_x^\alpha \sigma_Q'^{1-\alpha}\right)\\
& \qquad\qquad+\eta\frac{1}{\alpha-1}\log\tr\left(\rho_0^\alpha\sigma_Q'^{1-\alpha}\right)\\
&=(1-\eta)D_\alpha (\rho_{XQ}\|\rho_X\otimes \sigma_Q')+\eta D_\alpha(\rho_0\|\sigma_Q')\ 
\end{align*}
Here we used the concavity of $\log$ to obtain the inequality.
In particular, bounding $D_\alpha(\rho_{XQ}\|\rho_X\otimes\sigma_Q')$ by $\chi_\alpha(\rho_{XQ})$ using~\eqref{eq:nine}, we get
\begin{align*}
\chi_\alpha(\rho_{XQ}') \geq (1-\eta)\chi_\alpha(\rho_{XQ}) +\eta D_\alpha(\rho_0\|\mr{\rho_{XQ}'})\ .
\end{align*}
This is the claim~\eqref{eq:fourteen}.
\end{proof}

\subsection{Optimal $\cA$-ensembles for $\chi_\alpha$}
We now restrict the quantum states $\rho_x$ to be in some subset $\cA\subseteq \cS(\hil_Q)$. For a fixed set $\cA\subseteq\cS(\hil_Q)$, we define an $\cA$-ensemble to be an ensemble $\{p_x,\rho_x\}_{x}$ where $\rho_x\in\cA$ for all $x\in\cX$. An $\cA$-cq-state $\rho_{XQ} $ is a cq-state defined by an $\cA$-ensemble. Our main focus is on the $\alpha$-Holevo quantity, maximized over all $\cA$-ensembles (or equivalently all $\cA$-cq-states), that is, the quantity
\begin{align*}
\chi_\alpha^*(\cA)& := \max_{\{p_x,\rho_x\in\cA\}_x} \chi_\alpha(\{p_x,\rho_x\}_x)\ .
\end{align*}

We can show a maximal distance property similar to the one derived in~\cite{schumacher:capacityVSrelative} for the Holevo-quantity $\chi$ and the relative entropy $D$.
\begin{lemma}\label{lem:maximaldistanceprop}
Let $\cA\in\cS(\cH)$ be some set of states, and suppose the $\cA$-cq-state $\rho_{XQ}^*$ achieves the maximum of $\chi_{\alpha}$, that is, $\chi_\alpha(\rho_{XQ}^*)=\chi_\alpha^*(\cA)$. 
Then 
\begin{align*}
D_\alpha(\rho_0\|\mr{\rho_{XQ}^*}) \leq \capa_\alpha(\cA)\qquad\textrm{ for any state }\rho_0\in\cA\ .
\end{align*}
\end{lemma}
\begin{proof}
Assume that there exists a state $\rho_0\in\cA$ such that 
\begin{align}
D_\alpha(\rho_0\|\mr{\rho_{XQ}^*}) > \capa_\alpha(\cA)\ .\label{eq:dalphaassump}
\end{align}
Consider the state $\rho_{XQ}'=(1-\eta)\rho_{XQ}^{*}+\eta\proj{\bot}\otimes\rho_0$ for $0\leq\eta\leq 1$. Observe that this is a $\cA$-cq-state.  As $\eta\rightarrow 0$, we have $D_\alpha(\rho_0\|\mr{\rho_{XQ}'})\rightarrow D_\alpha(\rho_0\|\mr{\rho_{XQ}^*})$  by continuity. In particular, by~\eqref{eq:dalphaassump}, there is a value of $\eta$ such that
\begin{align*}
D_\alpha(\rho_0\|\mr{\rho_{XQ}'})>\chi_\alpha^*(\cA)\ .
\end{align*}
Combining this with~\eqref{eq:fourteen} leads to the contradiction
\begin{align*}
\chi_\alpha(\rho'_{XQ})>\chi_\alpha(\rho^{*}_{XQ})=\chi_\alpha^*(\cA)\ .
\end{align*}
\end{proof}

We are now ready to prove the following lemma. Note that~\eqref{eq:chiD} of Lemma~\ref{lem:evaluation} corresponds to the special case where 
$\cA=\{\Lambda(\rho) \mid \rho \in \states(\hin)\}$ is chosen as the set of potential output states of the channel~$\Lambda$.
\begin{lemma}\label{lem:generalevaluation}
Let $\cA\subseteq\cS(\cH)$  be a set of states and $\alpha\geq 1$. Then
\begin{align}
\min_{\sigma\in\cS(\cH)} \max_{\rho\in\cA} D_\alpha(\rho\|\sigma) \leq \chi_\alpha^*(\cA)\leq \min_{\sigma\in\cA}\max_{\rho\in\cA} D_\alpha(\rho\|\sigma) .\label{eq:chidgeneralized}
\end{align}
\end{lemma}
\begin{proof}
Consider an arbitrary $\cA$-cq-state $\rho_{XQ}$. We show that for any $\sigma \in \setA$, the quantity $\max_{\rho \in \setA} D_\alpha(\rho\|\sigma)$ is an upper bound on any quantity $\chi_\alpha(\rho_{XQ})$. Indeed, by~\eqref{eq:nine}, we have
\begin{align*}
\chi_\alpha(\rho_{XQ}) &\leq D_\alpha(\rho_{XQ}\|\rho_X\otimes\sigma)\\
&\leq \max_{\tilde{\rho}_{XQ}\textrm{ $\cA$-cq state}} D_\alpha(\tilde{\rho}_{XQ}\|\tilde{\rho}_X\otimes\sigma)\\
&=\frac{1}{\alpha-1}\log\max_{\{\tilde{p}_x,\tilde{\rho}_x \in \setA\}}\tr\left(\sum_x \tilde{p}_x\tilde{\rho}_x^\alpha\sigma^{1-\alpha}\right)\\
&=\frac{1}{\alpha-1}\log\max_{\rho \in \setA} \tr(\rho^\alpha\sigma^{1-\alpha})\\
&=\max_{\rho \in \setA} D_\alpha(\rho\|\sigma)\ .
\end{align*}
The upper bound in~\eqref{eq:chidgeneralized} follows from this by taking the minimum over $\sigma\in\cA$.

We know from Lemma~\ref{lem:maximaldistanceprop} that 
\begin{align}
\capa_\alpha(\cA) &\geq \max_{\rho \in \setA} D_\alpha(\rho\|\mr{\rho_{XQ}^*}) \ ,\label{eq:vx}
\end{align}
where $\rho_{XQ}^*$ is the state that achieves the optimum in~$\chi^*(\cA)$. Observe that $\mr{\rho_{XQ}^*}\in\cS(\cH)$ is  a state (but not necessarily an element of $\cA$). Therefore, the lower bound in~\eqref{eq:chidgeneralized} follows from~\eqref{eq:vx}.
\end{proof}

\section{Proof of Lemma~\ref{lem:errorexponentbound}}\label{sec:exponent}
\setcounter{section}{3}
We separate the two parts of the proof of this lemma, first addressing the general bound on the error exponent in terms of the regularized Holevo quantities.
\begin{proof}[Proof of part~\eqref{it:firstpartoperational}]
Consider a CPTPM $\Phi:\bounded(\cH_{in})\rightarrow\bounded(\cH_{out})$. 
Fix a set of states $\{\rho_x\}_{x=1}^{2^{nR}}\subset\states(\hin^{\otimes n})$ and a POVM $\{E_x\}_{x=1}^{2^{nR}}$ on $\hout^{\otimes n}$, and  consider the success probability $P_{succ}^{\Phi}(n,R)$ defined by~\eqref{eq:success}.
Let $\sigma_x=\Phi^{\otimes n}(\rho_x)$. Since $y\mapsto y^{1/\alpha}$ is operator monotone for $\alpha > 1$ (see e.g.,~\cite[Theorem V.1.9]{bathia:ma}), we have the operator inequality
\begin{align*}
\sigma_x=(\sigma_x^\alpha)^{1/\alpha}\leq \left(\sum_{x'} \sigma_{x'}^{\alpha}\right)^{1/\alpha}\ 
\end{align*}
for all $x$.
Inserting this into~\eqref{eq:success} gives
\begin{align*}
P_{succ}^{\Phi}(n,R)&\leq 2^{-nR}\sum_{x} \tr\left[E_x  \left(\sum_{x'} \sigma_{x'}^{\alpha}\right)^{1/\alpha}\right]\\
&\leq 2^{-nR}\tr\left[\left(\sum_x \sigma_x^{\alpha}\right)^{1/\alpha}\right]\\
&=2^{\frac{\alpha-1}{\alpha} (-nR+\chi_\alpha(\{2^{-nR},\sigma_x\}_x)}\\
&\leq 2^{\frac{\alpha-1}{\alpha} (-nR+\chi_\alpha^*(\Phi^{\otimes n}))}\ .
\end{align*}
Here we used the operator inequality $E_x\leq \id$ for POVM elements in the first step and the definition of $\chi_\alpha$
applied to the ensemble defined by $\{\sigma_x=\Phi^{\otimes n}(\rho_x)\}_x$ together with the uniform distribution on $\{1,\ldots,2^{nR}\}$. Since both the set of states and the POVM were arbitrary, the claim follows from definition~\eqref{eq:errorexponent}.
\end{proof}

\begin{proof}[Proof of part~\eqref{it:secondpartoperational}]
Substituting $\alpha=1/(s+1)$ for $-1<s<0$ gives
\begin{align}
\frac{\alpha-1}{\alpha}(R-\chi_\alpha(\{p_x,\Phi(\rho_x)\})&=-sR+E_0(s,\{p_x,\Phi(\rho_x)\})\ ,\label{eq:reform}
\end{align}
where 
\begin{align*}
E_0(s,\{p_x,\sigma_x\}_x) &:=s\cdot \chi_{1/(s+1)}(\{p_x,\sigma_x\}_x)\ .
\end{align*}
In~\cite[Lemma~3]{ogawa:converse}, it is shown that
for all families of states~$\{\sigma_x\}_x$, and all
$R>\max_{\{q_x\}_x} \chi(\{q_x,\sigma_x\}_x)$, we have
\begin{align*}
\exists t<0: -sR+\min_{\{q_x\}_x}E_0(s,\{q_x,\sigma_x\}_x)>0\qquad \forall s\in (t,0)\ ,
\end{align*}
where the maximum and minimum are over all probability distributions $\{q_x\}$.  Let $\{p_x,\rho_x\}_x$ be the ensemble which achieves the maximum  in the definition of $\chi^*(\Phi)$, and set $\sigma_x:=\Phi(\rho_x)$. With~\eqref{eq:reform} and the previous statement, we conclude that for all $R>\chi^*(\Phi)$, there exists $\beta>1$ such that 
\begin{align*}
\frac{\alpha-1}{\alpha} (R-\chi_\alpha^*(\Phi))>0\qquad\forall \alpha\in (1,\beta)\ .
\end{align*}
This is part~\eqref{it:secondpartoperational} of the claim since $\alpha>1$. 
\end{proof}

\section{An additional technical lemma}
\setcounter{section}{4}

\begin{lemma}\label{lem:majorize}
Let $A\geq 0$ and $B\geq 0$ be two positive semi-definite operators on $\Complex^{d \times d}$
with eigenvalues $\lambda^A=(\lambda^A_1, \ldots, \lambda^A_d)$ and $\lambda^B=(\lambda^B_1, \ldots, \lambda^B_d)$. Then
there exist
permutations $\pi_{\min},\pi_{\max} \in S_d$ such that
for all unitaries~$U$
$$
\sum_{j=1}^d \lambda^{A}_{\pi_{\min}(j)} \lambda^{B}_j \leq \Tr(UAU^\dagger B) \leq \sum_{j=1}^d \lambda^{A}_{\pi_{\max}(j)} \lambda^{B}_j.
$$
\end{lemma}
\begin{proof}

Let $v = (v_1,\ldots,v_d)$ be the vector of diagonal entries of the matrix $UAU^\dagger$ in a basis consisting of normalized eigenvectors of $B$. A  well-known result by Schur (see e.g.,~\cite[Theorem 4.3.26]{horn&johnson:ma}) states that  the vector of diagonal entries of a nonnegative matrix majorizes the vector of its eigenvalues. Applied to $UAU^\dagger$, we conclude that $v$ majorizes $\lambda^A$.
  Another classical theorem by Hardy, Littlewood and P\'olya (see e.g.,~\cite[Theorem 4.3.33]{horn&johnson:ma})  then shows that there exists a probability distribution~$P$ (depending on $U$) over the group of permutations $S_d$ such that
$$
(v_1,\ldots,v_d) = \sum_{\pi \in S_d} P(\pi) (\lambda^A_{\pi(1)},\ldots,\lambda^A_{\pi(d)}).
$$
The claim now follows
by observing that $\Tr(UAU^\dagger B) = v \cdot \lambda^B$, where $v \cdot \lambda^B$ denotes the Euclidean inner product of vectors $v$ and $\lambda^B$.
\end{proof}

\end{document}